\documentclass[runningheads]{llncs} % fleqn

\usepackage{verbatim}
\usepackage{amsmath}
\usepackage{amssymb}
\usepackage{stmaryrd} % semantic brackets
\usepackage{wasysym}
\usepackage{mathtools} % \coloneqq symbol
\usepackage{bm}
\usepackage{listings}
\usepackage{proof}
\usepackage{xcolor}
\usepackage{pifont}
\usepackage{multicol}
\usepackage{framed}
\usepackage{multirow}

\usepackage{underscore}
\usepackage{cite}

\usepackage{tikz}
\usetikzlibrary{arrows,automata,positioning}

\usepackage{ltltikz}

\usepackage{subfig}
\usepackage{wrapfig}

\usepackage[firstpage]{draftwatermark}
\SetWatermarkText{\hspace*{5in}\raisebox{8.2in}{\includegraphics[scale=0.1]{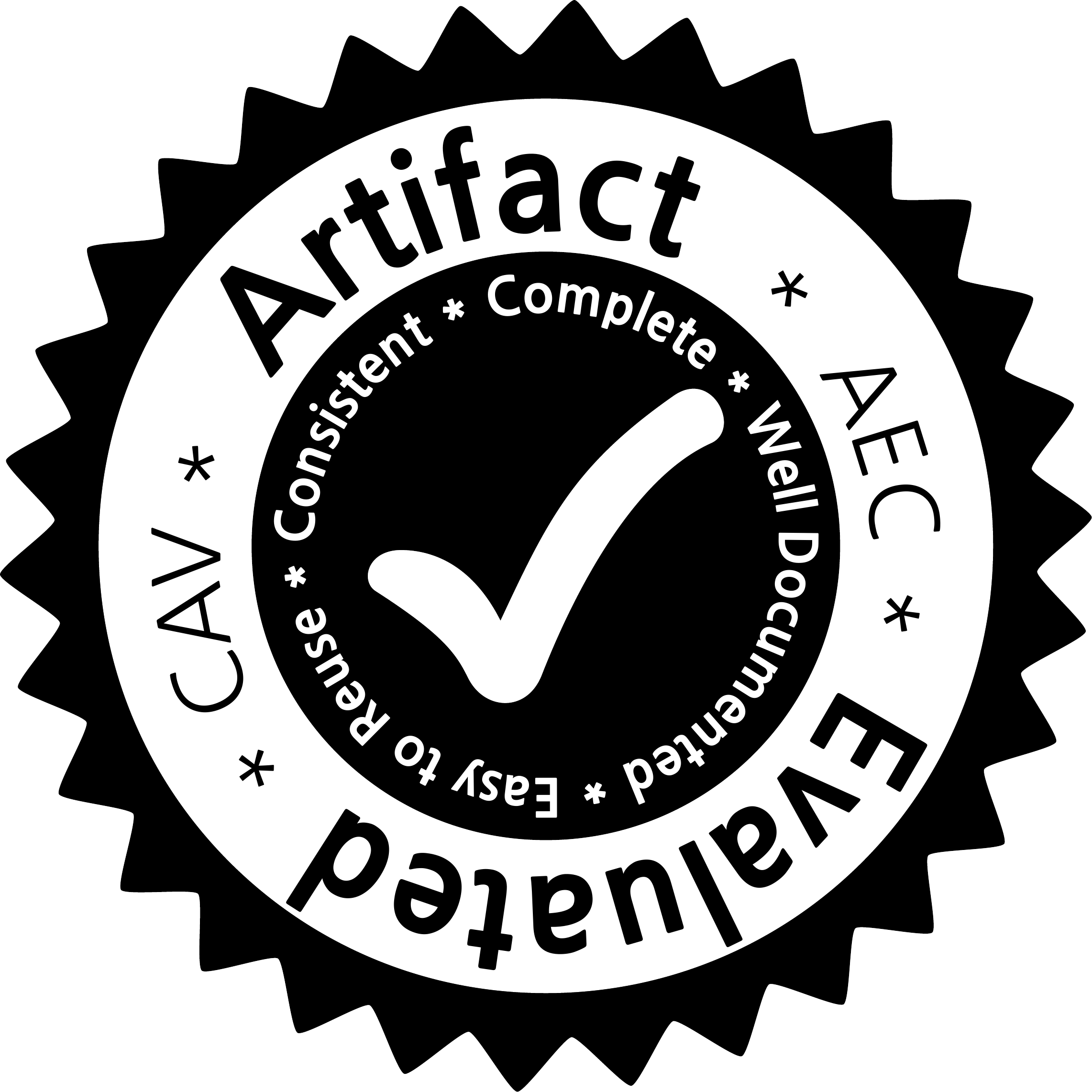}}}
\SetWatermarkAngle{0}

\usepackage{hyperref}

\newcommand{\tuple}[1]{\langle #1 \rangle}
\newcommand{\set}[1]{\{ #1 \}}
\newcommand{\ldot}{\mathpunct{.}}
\newcommand{\pow}[1]{2^{#1}}

\newcommand{\card}[1]{|{#1}|}

\newcommand{\twoexptime}{\textsc{2ExpTime}}

\newcommand{\bool}{\mathbb{B}}
\newcommand{\nat}{\mathbb{N}}
\newcommand{\Tr}{\mathit{Tr}}

\newcommand{\ltl}{\text{LTL}}
\newcommand{\hyperltl}{\text{HyperLTL}}

\newcommand{\mc}{\textsc{MCHyper}}
\newcommand{\abc}{\textsc{abc}}
\newcommand{\aiger}{\textsc{Aiger}}

\newcommand{\zthree}{\textsc{z3}}

\newcommand{\tsys}{\mathcal{S}}
\newcommand{\T}{S}
\newcommand{\st}{s}

\newcommand{\ucw}{\mathcal{A}}
\newcommand{\aut}{\mathcal{A}}
\newcommand{\acc}{\alpha}

\newcommand{\pathassign}{\Pi}
\newcommand{\pathvars}{\mathcal{V}}
\renewcommand{\models}{\vDash}
\newcommand{\nmodels}{\nvDash}
\newcommand{\ap}{\text{AP}}
\newcommand{\lang}{\mathcal{L}}
\newcommand{\zip}{\mathit{zip}}
\newcommand{\unzip}{\mathit{unzip}}

\newcommand{\U}{\Until}
\newcommand{\X}{\Next}
\newcommand{\G}{\Globally}
\newcommand{\F}{\Eventually}
\newcommand{\R}{\Release}
\newcommand{\W}{\WUntil}
\newcommand{\true}{\mathit{true}}
\newcommand{\false}{\mathit{false}}

\title{Verifying Hyperliveness\thanks{This work was supported in part by the Collaborative Research Center ``Foundations of Perspicuous Software Systems'' (TRR 248, 389792660), by the European Research Council (ERC) Grant OSARES (No. 683300), by Madrid Reg. Government project ``S2018/TCS-4339 (BLOQUES-CM)'', by EU H2020 project 731535 ``Elastest'' and by Spanish National Project ``BOSCO (PGC2018-102210-B-100)''.}}
\author{Norine Coenen\inst{1} \and Bernd Finkbeiner\inst{1} \and \\ C\'esar S\'anchez\inst{2} \and Leander Tentrup\inst{1}}
\authorrunning{N. Coenen, B. Finkbeiner, C. S\'anchez, L. Tentrup}
\institute{Reactive Systems Group, Saarland University \\ \texttt{lastname@react.uni-saarland.de} \and IMDEA Software Institute \\ \texttt{cesar.sanchez@imdea.org}}

\begin{document}

\maketitle

\begin{abstract}
HyperLTL is an extension of linear-time temporal logic for the specification of hyperproperties, i.e., temporal properties that relate multiple computation traces.
HyperLTL can express information flow policies as well as properties like symmetry in mutual exclusion algorithms or Hamming distances in error-resistant transmission protocols.
Previous work on HyperLTL model checking has focussed on the alternation-free fragment of HyperLTL, where verification reduces to checking a standard trace property over an appropriate self-composition of the system. 
The alternation-free fragment does, however, not cover general hyperliveness properties. 
Universal formulas, for example, cannot express the secrecy requirement that for every possible value of a secret variable there exists a computation where the value is different while the observations made by the external observer are the same.
In this paper, we study the more difficult case of hyperliveness properties expressed as  HyperLTL formulas with quantifier alternation.
We reduce existential quantification to strategic choice and show that synthesis algorithms can be used to eliminate the existential quantifiers automatically.
We furthermore show that this approach can be extended to reactive system synthesis, i.e., to automatically construct a reactive system that is guaranteed to satisfy a given HyperLTL formula.
\end{abstract}

%===============================================================================
\section{Introduction}
\label{sec:introduction}
%===============================================================================

HyperLTL~\cite{conf/post/ClarksonFKMRS14} is a temporal logic for \emph{hyperproperties}~\cite{journals/jcs/ClarksonS10}, i.e., for properties that relate multiple computation traces. 
Hyperproperties cannot be expressed in standard linear-time temporal logic (LTL), because LTL  can only express \emph{trace properties}, i.e., properties that characterize the correctness of individual computations. 
Even branching-time temporal logics like CTL and CTL$^*$, which quantify over computation paths, cannot express hyperproperties, because quantifying over a second path automatically means that the subformula can no longer refer to the previously quantified path.
HyperLTL addresses this limitation with quantifiers over trace variables, which allow the subformula to refer to all previously chosen traces. 
For example, \emph{noninterference}~\cite{conf/sp/GoguenM82a} between a secret input $h$ and a public output $o$ can be specified in HyperLTL by requiring that all pairs of traces $\pi$ and $\pi'$ that always have the same inputs except for $h$ (i.e., all inputs in $I\setminus\{h\}$ are equal on $\pi$ and $\pi'$) also have the same output $o$ at all times:
\[
\forall\pi.\forall\pi'.~ \G \big(\!\!\!\bigwedge_{i\in I\setminus \{h\}}\! i_\pi = i_{\pi'}\big) ~\Rightarrow~ \G\, (o_\pi = o_{\pi'})
\]
This formula states that a change in the secret input $h$ alone cannot cause any difference in the output $o$. 

For certain properties of interest, the additional expressiveness of HyperLTL comes at no extra cost when considering the model checking problem. 
To check a property like noninterference, which only has universal trace quantifiers, one simply builds the self-composition of the system, which provides a separate copy of the state variables for each trace. 
Instead of quantifying over all pairs of traces, it then suffices to quantify over individual traces of the self-composed system, which can be done with standard LTL. 
Model checking universal formulas is NLOGSPACE-complete in the size of the system and PSPACE-complete in the the size of the formula, which is precisely the same complexity as for LTL.

Universal HyperLTL formulas suffice to express hypersafety properties like noninterference, but not hyperliveness properties that require, in general, quantifier alternation. 
A prominent example is \emph{generalized noninterference} (GNI)~\cite{conf/sp/McCullough88}, which can be expressed as the following HyperLTL formula:
\[
\forall\pi.\forall\pi'.\exists\pi''.~\G\, (h_\pi = h_{\pi''}) ~\wedge~ \G\, (o_{\pi'} = o_{\pi''})
\]
This formula requires that for every pair of traces $\pi$ and $\pi'$, there is a third trace $\pi''$ in the system that agrees with $\pi$ on $h$ and with $\pi'$ on $o$. 
The existence of an appropriate trace $\pi''$ ensures that in $\pi$ and $\pi'$, the value of $o$ is not determined by the value of $h$.
Generalized noninterference stipulates that low-security outputs may not be altered by the injection of high-security inputs, while permitting nondeterminism in the low-observable behavior.
The existential quantifier is needed to allow this nondeterminism.
GNI is a hyperliveness property~\cite{journals/jcs/ClarksonS10} even though the underlying LTL formula is a safety property. 
The reason for that is that we can extend any set of traces that violates GNI into a set of traces that satisfies GNI, by adding, for each offending pair of traces $\pi, \pi'$, an appropriate trace $\pi''$.

Hyperliveness properties also play an important role in applications beyond security.
For example, \emph{robust cleanness}~\cite{conf/esop/DArgenioBBFH17} specifies that significant differences in the output behavior are only permitted after significant differences in the input:
\[
\forall {\pi}.\forall {\pi'}.\exists{\pi''}.\
\G \big(i_{\pi'} = i_{\pi''}\big) 
~\wedge~
\big(\hat{d}(o_{\pi},o_{\pi''})\leq \kappa_o~\mathcal W~ \hat{d}(i_{\pi},i_{\pi''})>\kappa_i \big)
\]
The differences are measured by a distance function $\hat{d}$ and compared to constant thresholds $\kappa_i$ for the input and $\kappa_o$ for the output.
The formula specifies the existence of a trace $\pi''$ that globally agrees with $\pi'$ on the input and where the difference in the output $o$ between $\pi$ and $\pi''$ is bounded by $\kappa_o$, unless the difference in the input $i$ between $\pi$ and $\pi''$ was greater than $\kappa_i$.  
Robust cleanness, thus, forbids unexpected jumps in the system behavior that are, for example, due to software doping, while allowing for behavioral differences due to nondeterminism.

With quantifier alternation, the model checking problem becomes much more difficult. 
Model checking HyperLTL formulas of the form $\forall^* \exists^* \varphi$, where $\varphi$ is a quantifier-free formula, is PSPACE-complete in the size of the system and EXPSPACE-complete in the formula. 
The only known model checking algorithm replaces the existential quantifier with the negation of a universal quantifier over the negated subformula; but this requires a complementation of the system behavior, which is completely impractical for realistic systems.

In this paper, we present an alternative approach to the verification of hyperliveness properties.
We view the model checking problem of a formula of the form $\forall \pi. \exists \pi'.\ \varphi$ as a game between the $\forall$-player and the $\exists$-player. 
While the $\forall$-player moves through the state space of the system building trace $\pi$, the $\exists$-player must match each move in a separate traversal of the state space resulting in a trace $\pi'$ such that the pair $\pi, \pi'$ satisfies $\varphi$. 
Clearly, the existence of a winning strategy for the $\exists$-player implies that $\forall \pi. \exists \pi'.\ \varphi$ is satisfied.
The converse is not necessarily true: 
Even if there always is a trace $\pi'$ that matches the universally chosen trace $\pi$, the $\exists$-player may not be able to construct this trace, because she only knows about the choices made by the $\forall$-player in the finite prefix of $\pi$ that has occurred so far, and not the choices that will be made by the $\forall$-player in the infinite future.
We address this problem by introducing \emph{prophecy variables} into the system. 
Without changing the behavior of the system, the prophecy variables give the $\exists$-player the information about the future that is needed to make the right choice after seeing only the finite prefix. 
Such prophecy variables can be provided manually by the user of the model checker to provide a lookahead on future moves of the $\forall$-player.

This game-theoretic approach provides an opportunity for the user to reduce the complexity of the model checking problem: 
If the user provides a strategy for the $\exists$-player, then the problem reduces to the cheaper model checking problem for universal properties.
We show that such strategies can also be constructed automatically using synthesis.
Beyond model checking, the game-theoretic approach also provides a method for the synthesis of systems that satisfy a conjunction of hypersafety and hyperliveness properties. 
Here, we do not only synthesize the strategy, but also construct the system itself, i.e., the game graph on which the model checking game is played. 
While the synthesis from $\forall^* \exists^*$ hyperproperties is known to be undecidable in general, we show that the game-theoretic approach can naturally be integrated into bounded synthesis, which checks for the existence of a correct system up to a bound on the number of states.

\paragraph{\bf Related Work. }
While the verification of general HyperLTL formulas has been studied before~\cite{conf/post/ClarksonFKMRS14,conf/cav/FinkbeinerRS15,conf/cav/FinkbeinerHT18}, there has been, so far, no practical model checking algorithm for HyperLTL formulas with quantifier alternation. 
The existing algorithm involves a complementation of the system automaton, which results in an exponential blow-up of the state space~\cite{conf/cav/FinkbeinerRS15}. 
The only existing model checker for HyperLTL, $\mc$~\cite{conf/cav/FinkbeinerRS15}, was therefore, so far, limited to the alternation-free fragment. 
Although some hyperliveness properties lie in this fragment, quantifier alternation is needed to express general hyperliveness properties like GNI. 
In this paper, we present a technique to model check these hyperliveness properties and extend $\mc$ to formulas with quantifier alternation.

The situation is similar in the area of reactive synthesis. 
There is a synthesis algorithm that automatically constructs implementations from HyperLTL specifications~\cite{conf/cav/FinkbeinerHLST18} using the bounded synthesis approach~\cite{journals/sttt/FinkbeinerS13}. 
This algorithm is, however, also only applicable to the alternation-free fragment of HyperLTL. 
In this paper, we extend the bounded synthesis approach to HyperLTL formulas with quantifier alternation.
Beyond the model checking and synthesis problems, the satisfiability~\cite{conf/concur/FinkbeinerH16,conf/cav/FinkbeinerHS17,conf/atva/FinkbeinerHH18} and monitoring~\cite{conf/rv/FinkbeinerHST17,conf/tacas/FinkbeinerHST18,conf/tacas/HahnST19} problems of HyperLTL have also been studied in the past. 

For certain information-flow security policies, there are verification techniques that use methods related to our model checking and synthesis algorithms. 
Specifically, the self-composition technique~\cite{conf/lfcs/BartheCK13,conf/csfw/BartheDR04}, a construction based on the product of copies of a system, has been tailored for various trace-based security definitions~\cite{journals/jcs/DSouzaHRS11,conf/csfw/HuismanWS06,journals/entcs/MeydenZ07}. 
Unlike our algorithms, these techniques focus on specific information-flow policies, not on a general logic like HyperLTL.

The use of prophecy variables~\cite{journals/tcs/AbadiL91} to make information about the future accessible is a known technique in the verification of trace properties. 
It is, for example, used to establish simulation relations between automata~\cite{journals/iandc/LynchV95} or in the verification of CTL$^*$ properties~\cite{conf/cav/CookKP15}.  

In our game-theoretic view on the model checking problem for $\forall^* \exists^*$ hyperproperties the $\exists$-player has an infinite lookahead.
There is some work on \emph{finite} lookahead on trace languages~\cite{conf/icalp/Klein015}. 
We use the idea of finite lookahead as an approximation to construct existential strategies and give a novel synthesis construction for strategies with delay based on bounded synthesis~\cite{journals/sttt/FinkbeinerS13}.  

%===============================================================================
\section{Preliminaries}
\label{sec:preliminaries}
%===============================================================================

For tuples $\vec{x} \in X^n$ and $\vec{y} \in X^m$ over set $X$, we use $\vec{x} \cdot \vec{y} \in X^{n+m}$ to denote the concatenation of $\vec{x}$ and $\vec{y}$.
Given a function $f\colon X \to Y$ and a tuple $\vec{x} \in X^n$, we define by $f \circ \vec{x} \in Y^n$ the tuple $(f(\vec{x}[1]), \dots, f(\vec{x}[n]))$. 
Let $\ap$ be a finite set of atomic propositions and let $\Sigma = 2^\ap$ be the corresponding alphabet.
A \emph{trace} $t \in \Sigma^\omega$ is an infinite sequence of elements of $\Sigma$.
We denote a set of traces by $\Tr \subseteq \Sigma^\omega$.
We define $t[i,\infty]$ to be the suffix of $t$ starting at position $i \geq 0$.

%-------------------------------------------------------------------------------
\paragraph{HyperLTL.}
%-------------------------------------------------------------------------------

$\hyperltl$~\cite{conf/post/ClarksonFKMRS14} is a temporal logic for specifying hyperproperties.
It extends $\ltl$ by quantification over trace variables $\pi$ and a method to link atomic propositions to specific traces.
Let $\pathvars$ be an infinite set of trace variables.
Formulas in $\hyperltl$ are given by the grammar
\begin{align*}
\varphi &{}\Coloneqq \forall\pi\ldot\varphi \mid \exists\pi\ldot\varphi \mid \psi \enspace, \text{ and}\\
\psi &{}\Coloneqq a_\pi \mid \neg\psi \mid \psi\lor\psi \mid \X\psi \mid \psi\U\psi \enspace,
\end{align*}
where $a \in \ap$ and $\pi \in \pathvars$.
We allow the standard boolean connectives $\wedge$, $\rightarrow$, $\leftrightarrow$ as well as the derived $\ltl$ operators release $\varphi \R \psi \equiv \neg(\neg\varphi \U \neg\psi)$, eventually $\F \varphi \equiv \true\ \U\ \varphi$, globally $\G \varphi \equiv \neg \F \neg \varphi$, and weak until $\varphi \W \psi \equiv \G \varphi \lor (\varphi\ \U\ \psi)$.

We call a $\mathcal{Q}^+ \mathcal{Q}'^+ \varphi\, \hyperltl$ formula (for $\mathcal{Q},\mathcal{Q}' \in \{\forall, \exists\}$ and quantifier-free formula $\varphi$) \emph{alternation-free} iff $\mathcal{Q} = \mathcal{Q}'$.
Further, we say that $\mathcal{Q}^+ \mathcal{Q}'^+ \varphi$ has \emph{one quantifier alternation} (or lies in the \emph{one-alternation fragment}) iff $\mathcal{Q} \neq \mathcal{Q}'$. 

The semantics of $\hyperltl$ is given by the satisfaction relation $\models_\Tr$ over a set of traces $\Tr \subseteq \Sigma^\omega$.
We define an assignment $\pathassign : \pathvars \to \Sigma^\omega$ that maps trace variables to traces.
$\pathassign[\pi \mapsto t]$ updates $\pathassign$ by assigning variable $\pi$ to trace $t$. 
\begin{equation*}
\begin{array}{lll}
  \pathassign,i \models_\Tr a_\pi       \qquad \qquad & \text{iff } & a \in \pathassign(\pi)[i] \\
  \pathassign,i \models_\Tr \neg \varphi              & \text{iff } & \pathassign,i \nmodels_\Tr \varphi \\
  \pathassign,i \models_\Tr \varphi \lor \psi         & \text{iff } & \pathassign,i \models_\Tr \varphi \text{ or } \pathassign,i \models_\Tr \psi \\
  \pathassign,i \models_\Tr \X \varphi                & \text{iff } & \pathassign,i+1 \models_\Tr \varphi \\
  \pathassign,i \models_\Tr \varphi\U\psi             & \text{iff } & \exists j \geq i \ldot \pathassign, j \models_\Tr \psi \land \forall i \leq k < j \ldot \pathassign,k \models_\Tr \varphi \\
  \pathassign,i \models_\Tr \exists \pi \ldot \varphi & \text{iff } & \text{there is some } t \in \Tr \text{ such that } \pathassign[\pi \mapsto t],i \models_\Tr \varphi\\
  \pathassign,i \models_\Tr \forall \pi \ldot \varphi & \text{iff } & \text{for all } t \in \Tr \text{ it holds that } \pathassign[\pi \mapsto t],i \models_\Tr \varphi
\end{array}
\end{equation*}
We write $\Tr \models \varphi$ for $\set{},0 \models_\Tr \varphi$ where $\set{}$ denotes the empty assignment.

Every hyperproperty is an intersection of a hypersafety and a hyperliveness property~\cite{journals/jcs/ClarksonS10}. 
A \emph{hypersafety} property is one where there is a finite set of finite traces that is a bad prefix, i.e., that cannot be extended into a set of traces that satisfies the hypersafety property. 
A \emph{hyperliveness} property is a property where every finite set of finite traces can be extended to a possibly infinite set of infinite traces such that the resulting trace set satisfies the hyperliveness property. 

%-------------------------------------------------------------------------------
\paragraph{Transition Systems.}
%-------------------------------------------------------------------------------

We use transition systems as a model of computation for reactive systems.
Transition systems consume sequences over an input alphabet by transforming their internal state in every step.
Let $I$ and $O$ be a finite set of input and output propositions, respectively, and let $\Upsilon = \pow{I}$ and $\Gamma = \pow{O}$ be the corresponding finite alphabets.
A $\Gamma$-labeled $\Upsilon$-\emph{transition system} $\tsys$ is a tuple $\tuple{\T,\st_0,\tau,l}$, where
$\T$ is a finite set of states,
$\st_0 \in \T$ is the designated initial state,
$\tau \colon \T \times \Upsilon \to \T$ is the transition function, and
$l \colon \T \to \Gamma$ is the state-labeling function.
We write $\st \xrightarrow{\upsilon} \st'$ or $(\st, \upsilon, \st') \in \tau$ if $\tau(\st,\upsilon) = \st'$.
We generalize the transition function to sequences over $\Upsilon$ by defining $\tau^* \colon \Upsilon^* \to \T$ recursively as $\tau^*(\epsilon) = \st_0$ and $\tau^*(\upsilon_0 \cdots \upsilon_{n-1} \upsilon_n) = \tau(\tau^*(\upsilon_0 \cdots \upsilon_{n-1}), \upsilon_n)$ for $\upsilon_0 \cdots \upsilon_{n-1} \upsilon_n \in \Upsilon^+$.
Given an infinite word $\upsilon = \upsilon_0 \upsilon_1 \ldots \in \Upsilon^\omega$, the transition system produces an infinite sequence of outputs $\gamma = \gamma_0 \gamma_1 \gamma_2 \ldots \in \Gamma^\omega$, such that $\gamma_i = l(\tau^*(\upsilon_0 \ldots \upsilon_{i-1}))$ for every $i \geq 0$.
The resulting \emph{trace} $\rho$ is $(\upsilon_0 \cup \gamma_0) (\upsilon_1 \cup \gamma_1) \ldots \in \Sigma^\omega$ where we have $AP = I \cup O$.
The set of traces generated by $\tsys$ is denoted by $\mathit{traces}(\tsys)$.
Furthermore, we define $\varepsilon = \tuple{\{s\},s,\tau_\varepsilon,l_\varepsilon}$ as the transition system over $I = O = \emptyset$ that has only a single trace, that is $\mathit{traces}(\varepsilon) = \{\emptyset^\omega\}$. 
For this transition system, $\tau_\varepsilon(s, \emptyset) = s$ and $l_\varepsilon(s) = \emptyset$.
Given two transition systems $\tsys = \tuple{\T,\st_0,\tau,l}$ and $\tsys' = \tuple{\T',\st'_0,\tau',l'}$, we define $\tsys \times \tsys' = \tuple{\T \times \T',(\st_0, \st'_0),\tau'',l''}$ as the $\Gamma^2$-labeled $\Upsilon^2$-transition system where 
$\tau''((\st, \st'), (\upsilon, \upsilon')) = (\tau(\st, \upsilon), \tau'(\st', \upsilon'))$ and 
$l''((\st, \st')) = (l(\st), l'(\st'))$. 
A transition system $\tsys$ satisfies a general HyperLTL formula $\varphi$, if, and only if, $\mathit{traces}(\tsys) \models \varphi$. 

%-------------------------------------------------------------------------------
\paragraph{Automata.}
%-------------------------------------------------------------------------------

An alternating parity automaton $\aut$ over a finite alphabet $\Sigma$ is a tuple $\tuple{Q,q_0,\delta,\acc}$, where
$Q$ is a finite set of states,
$q_0 \in Q$ is the designated initial state,
$\delta \colon Q \times \Sigma \to \bool^+(Q)$ is the transition function, and
$\acc \colon Q \to C$ is a function that maps states of $\aut$ to a finite set of colors $C \subset \nat$.
For $C = \set{0,1}$ and $C = \set{1,2}$, we call $\aut$ a co-B\"uchi and B\"uchi automaton, respectively, and we use the sets $F \subseteq Q$ and $B \subseteq Q$ to represent the rejecting ($C=1$) and accepting ($C=2$) states in the respective automaton (as a replacement of the coloring function $\acc$).
A safety automaton is a B\"uchi automaton where every state is accepting.
The transition function $\delta$ maps a state $q \in Q$ and some $a \in \Sigma$ to a positive Boolean combination of successor states $\delta(q, a)$.
An automaton is \emph{non-deterministic} or \emph{universal} if $\delta$ is purely disjunctive or conjunctive, respectively.

A run of an alternating automaton is a $Q$-labeled tree.  
A tree $T$ is a subset of $\nat_{>0}^*$ such that for every node $n \in \nat^*_{>0}$ and every positive integer $i \in \nat_{>0}$, if $n \cdot i \in T$ then (i) $n \in T$ (i.e., $T$ is prefix-closed), and (ii) for every $0 < j < i$, $n \cdot j \in T$.
The root of $T$ is the empty sequence $\epsilon$ and for a node $n \in T$, $| n |$ is the length of the sequence $n$, in other words, its distance from the root.
A run of $\mathcal{A}$ on an infinite word $\rho \in \Sigma^\omega$ is a $Q$-labeled tree $(T,r)$ such that $r(\epsilon) = q_0$ and for every node $n \in T$ with children $n_1, \dots, n_k$ the following holds: 
$1 \leq k \leq |Q|$ and $\{r(n_1), \dots, r(n_k)\} \models \delta(q,\rho[i])$, where $q = r(n)$ and $i = |n|$.
A path is accepting if the highest color appearing infinitely often is even.
A run is accepting if all its paths are accepting.
The language of $\aut$, written $\lang(\aut)$, is the set $\set{\rho \in \Sigma^\omega \mid \aut \text{ accepts } \rho}$.
A transition system $\tsys$ is accepted by an automaton $\aut$, written $\tsys \models \aut$, if $\mathit{traces}(\tsys) \subseteq \lang(\aut)$.

\vspace{-5pt}
%-------------------------------------------------------------------------------
\paragraph{Strategies.}
%-------------------------------------------------------------------------------

Given two disjoint finite alphabets $\Upsilon$ and $\Gamma$, a strategy $\sigma \colon \Upsilon^* \to \Gamma$ is a mapping from finite histories of $\Upsilon$ to $\Gamma$.
A transition system $\tsys = \tuple{\T,\st_0,\tau,l}$ \emph{generates} the strategy $\sigma$ if $\sigma(\vec\upsilon) = l(\tau^*(\vec\upsilon))$ for every $\vec\upsilon \in \Upsilon^*$. 
A strategy $\sigma$ is called \emph{finite-state} if there exists a transition system that generates $\sigma$.

In the following, we use finite-state strategies to modify the inputs of transition systems.
Let $\tsys = \tuple{\T,\st_0,\tau,l}$ be a transition system over input and output alphabets $\Upsilon$ and $\Gamma$ and let $\sigma \colon (\Upsilon')^* \to \Upsilon$ be a finite-state strategy.
Let $\tsys' = \tuple{\T',\st_0',\tau',l'}$ be the transition system implementing $\sigma$, then $\tsys\ ||\ \sigma = \tsys \ ||\ \tsys'$ is the transition system $\tuple{\T \times \T',(\st_0, \st_0'),\tau^{||},l^{||}}$ where
$\tau^{||} \colon (\T \times \T') \times \Upsilon' \to  (\T \times \T')$ is defined as $\tau^{||}((\st,\st'), \upsilon') = (\tau(\st, l'(\st')), \tau'(\st', \upsilon'))$ and
$l^{||} \colon (\T \times \T') \to \Gamma$ is defined as $l^{||}(\st,\st') = l(\st)$ for every $\st \in \T$, $\st' \in \T'$, and $\upsilon' \in \Upsilon'$.

\vspace{-5pt}
%-------------------------------------------------------------------------------
\paragraph{Model Checking HyperLTL.}
\label{sec:alt-freeMC}
%-------------------------------------------------------------------------------

We recap the model checking of universal HyperLTL formulas.
This case, as well as the dual case of only existential quantifiers, is well-understood and, in fact, efficiently implemented in the model checker $\mc$~\cite{conf/cav/FinkbeinerRS15}.
The principle behind the model checking approach is \emph{self-composition}, where we check a standard trace property on a composition of an appropriate number of copies of
the given system.

Let $\zip$ denote the function that maps an $n$-tuple of sequences to a single sequence of $n$-tuples, for example, $\zip([1, 2, 3], [4, 5, 6]) = [(1, 4), (2, 5), (3, 6)]$, and let $\unzip$ denote its inverse.
Given $\tsys = \tuple{\T,\st_0,\tau,l}$, the $n$-fold self-composition of $\tsys$ is the transition system $\tsys^n = \tuple{\T^n, \vec{\st_0'}, \tau_n, l_n}$, where 
$\vec{\st_0'} \coloneqq (\st_0,\dots,\st_0) \in \T^n$,
$\tau_n(\vec{\st}, \vec{\upsilon}) \coloneqq \tau \circ \zip(\vec{\st}, \vec{\upsilon})$ and $l_n(\vec{\st}) \coloneqq l \circ \vec{\st}$ for every $\vec{\st} \in \T^n$ and $\vec{\upsilon} \in \Upsilon^n$.
If $\mathit{traces}(\tsys)$ is the set of traces generated by $\tsys$, then $\set{ \zip(\rho_1,\dots,\rho_n) \mid \rho_1,\dots,\rho_n \in \mathit{traces}(\tsys) }$ is the set of traces generated by $\tsys^n$.
We use the notation $\zip(\varphi, \pi_1, \pi_2, \ldots, \pi_n)$ for some HyperLTL formula $\varphi$ to combine the trace variables $\pi_1, \pi_2, \ldots, \pi_n$ (occurring free in $\varphi$) into a fresh trace variable $\pi^*$.

\vspace{-5pt}
\begin{theorem}[Self-composition for universal HyperLTL formulas~\cite{conf/cav/FinkbeinerRS15}]
For a transition system $\tsys$ and a HyperLTL formula of the form $\forall \pi_1. \forall \pi_2. \ldots \forall \pi_n.\ \varphi$ it holds that $\tsys \models \forall \pi_1. \forall \pi_2. \ldots \forall \pi_n.\ \varphi$ iff $\tsys^n \models \forall \pi^*.\ \zip(\varphi,\pi_1,\pi_2,\ldots,\pi_n)$.
\end{theorem}  
\vspace{-5pt}
\begin{theorem}[Complexity of model checking universal formulas~\cite{conf/cav/FinkbeinerRS15}]
The model checking problem for universal HyperLTL formulas is PSPACE-complete in the size of the formula and NLOGSPACE-complete in the size of the transition system.
\end{theorem}
The complexity of verifying universal HyperLTL formulas is exactly the same as the complexity of verifying LTL formulas.
For HyperLTL formulas with quantifier alternations, the model checking problem is significantly more difficult.
\vspace{-5pt}
\begin{theorem}[Complexity of model checking formulas with one quantifier alternation~\cite{conf/cav/FinkbeinerRS15}]
The model checking problem for HyperLTL formulas with one quantifier alternation is in EXPSPACE in the size of the formula and in PSPACE in the size of the transition system.
\end{theorem}
One way to circumvent this complexity is to fix the existential choice and strengthen the formula to the universal fragment~\cite{conf/cav/FinkbeinerRS15,conf/esop/DArgenioBBFH17,conf/cav/FinkbeinerHLST18}.
While avoiding the complexity problem, this transformation requires deep knowledge of the system, is prone to errors, and cannot be verified automatically as the problem of checking implications becomes undecidable~\cite{conf/concur/FinkbeinerH16}.
In the following section, we present a technique that circumvents the complexity problem while still inheriting strong correctness guarantees.
Further, we provide a method that can, under certain restrictions, derive a strategy for the existential choice automatically.

\vspace{-5pt}
%===============================================================================
\section{Model Checking with Quantifier Alternations}
\label{sec:modelchecking}
%===============================================================================
\vspace{-5pt}
%-------------------------------------------------------------------------------
\subsection{Model Checking with Given Strategies}
\label{sec:givenstrat}
%-------------------------------------------------------------------------------

Our first goal is the verification of HyperLTL formulas with one quantifier alternation, i.e., formulas of the form 
$\forall^* \exists^* \varphi$  or $\exists^* \forall^* \varphi$, where $\varphi$ is a quantifier-free formula. 
Note that the presented techniques can, similar to skolemization, be extended to more than one quantifier alternation. 
Quantifier alternation introduces dependencies between the quantified traces. 
In a $\forall^*\exists^* \varphi$ formula, the choices of the existential quantifiers depend on the choices of the universal quantifiers preceding them. 
In a formula of the form $\exists^* \forall^* \varphi$, however, there has to be a single choice for the existential quantifiers that works for all choices of the universal quantifiers. 
In this case, the existentially quantified variables do not depend on the universally quantified variables. 
Hence, the witnesses for the existential quantifiers are traces rather than functions that map tuples of traces to traces.
As established above, the model checking problem for HyperLTL formulas with quantifier alternation is known to be significantly more difficult than the model checking problem for universal formulas.

Our verification technique for formulas with quantifier alternation is to substitute strategic choice for existential choice. As discussed in the introduction, the existence of a strategy implies the existence of a trace. 

\begin{theorem}[Substituting Strategic Choice for Existential Choice]
  \label{theo:stratex}
  Let $\tsys$ be a transition system over input alphabet $\Upsilon$. \\
  It holds that $\tsys \models \forall \pi_1 \forall \pi_2 \ldots \forall \pi_n.\ \exists \pi_1' \exists \pi_2' \ldots\exists \pi_m'.\  \varphi$ if there is a strategy $\sigma: (\Upsilon^n)^* \rightarrow \Upsilon^m$ such that $\tsys^n \times (\tsys^m \ ||\ \sigma) \models \forall \pi^*. \zip(\varphi, \pi_1, \pi_2, \ldots \pi_n, \pi_1', \pi'_2, \ldots, \pi'_m)$. \\
  It holds that  $\tsys \models \exists \pi_1 \exists \pi_2 \ldots \exists \pi_m.\ \forall \pi_1' \forall \pi_2' \ldots\forall \pi_n'.\  \varphi$ if there is a strategy $\sigma: (\Upsilon^0)^* \rightarrow \Upsilon^m$ such that $ (\tsys^m \ ||\ \sigma) \times \tsys^n \models \forall \pi^*. \zip(\varphi, \pi_1, \pi_2, \ldots \pi_m, \pi_1', \pi'_2, \ldots, \pi'_n)$.
\end{theorem}
\begin{proof}
  Let $\sigma$ be such a strategy, then we define a witness for the existential trace quantifiers $\exists \pi_1'\exists\pi_2'\ldots\exists \pi_m'$ as the sequence of inputs $\upsilon = \upsilon_0 \upsilon_1 \ldots \in (\Upsilon^m)^\omega$ such that $\upsilon_i = \sigma(\upsilon'_0 \upsilon'_1 \ldots \upsilon'_{i-1})$ for every $i \geq 0$ and every $\upsilon'_i \in \Upsilon^n$;
analogously, we define a witness for the existential trace quantifiers $\exists \pi_1\exists\pi_2\ldots\exists \pi_m$ as the sequence of inputs  $\upsilon = \upsilon_0 \upsilon_1 \ldots \in (\Upsilon^m)^\omega$ such that $\upsilon_i = \sigma(\upsilon'_0 \upsilon'_1 \ldots \upsilon'_{i-1})$ for every $i \geq 0$ and every $\upsilon'_i \in \Upsilon^0$.
  \qed
\end{proof} 
An application of the theorem reduces the verification problem of a HyperLTL formula with one quantifier alternation to the verification problem of a universal $\hyperltl$ formula. 
If a sufficiently small strategy can be found, the reduction in complexity is substantial:
\begin{corollary}[Model checking with Given Strategies]
The model checking problem for HyperLTL formulas with one quantifier alternation and given strategies for the existential quantifiers is in PSPACE in the size of the formula and NLOGSPACE in the size of the product of the strategy and the system.
\end{corollary}
Note that the converse of Theorem~\ref{theo:stratex} is not in general true. The satisfaction of a $\forall^* \exists^* \, \hyperltl$ formula does not imply the existence of a strategy, because  at any given point in time the strategy only knows about a finite prefix of the universally quantified traces. 
Consider the formula $\forall \pi \exists \pi'. \Next a_\pi \leftrightarrow a_{\pi'}$ and a system that can produce arbitrary sequences of $a$ and $\neg a$. 
Although the system satisfies the formula, it is not possible to give a strategy that allows us to prove this fact. 
Whatever choice our strategy makes, the next move of the $\forall$-player can make sure that the strategy's choice was wrong. 
In the following, we present a method that addresses this problem.

%-------------------------------------------------------------------------------
\subsubsection{Prophecy Variables.}
\label{sec:prophecy}
%-------------------------------------------------------------------------------

A classic technique for resolving future dependencies is the introduction of \emph{prophecy variables}~\cite{journals/tcs/AbadiL91}. 
Prophecy variables are auxiliary variables that are added to the system without affecting the behavior of the system. 
Such variables can be used to make predictions about the future.

We use prophecy variables to define strategies that depend on the future. 
In the example discussed above, $\forall \pi \exists \pi'. \Next a_\pi \leftrightarrow a_{\pi'}$, the choice of the value of $a_{\pi'}$ in the first position depends on the value of $a_\pi$ in the second position. 
We introduce a prophecy variable $p$ that predicts in the first position whether $a_\pi$ is true in the second position.
With the prophecy variable, there exists a strategy that correctly assigns the value of $p$ whenever the prediction is correct: 
The strategy chooses to set $a_{\pi'}$ if, and only if, $p$ holds. 

Technically, the proof technique introduces a set of fresh input variables $P$ into the system. 
For a $\Gamma$-labeled $\Upsilon$-transition system $\tsys = \tuple{\T,\st_0,\tau,l}$, we define the $\Gamma$-labeled $(\Upsilon \cup P)$-transition system $\tsys^P = \tuple{\T,\st_0,\tau^P,l}$ including the inputs $P$ where $\tau^P \colon \T \times (\Upsilon \cup P) \to \T$. 
For all $\st \in \T$ and $\upsilon^P \in \Upsilon \cup P$, $\tau^P(\st, \upsilon^P) = \tau(\st, \upsilon)$ for $\upsilon \in \Upsilon$ obtained by removing the variables in $P$ from $\upsilon^P$ (i.e., $\upsilon =_{\setminus P} \upsilon^P$). 
Moreover, the proof technique modifies the specification so that the original property only needs to be satisfied if the prediction is actually correct. 
We obtain the modified specification
$\forall \pi \exists \pi'. (p_{\pi} \leftrightarrow \Next a_\pi)\ \rightarrow\ (\Next a_\pi \leftrightarrow a_{\pi'})$
in our example.
The following theorem describes the general technique for one prophecy variable.

\begin{theorem}[Model checking with Prophecy Variables]
  For a transition system $\tsys$ and a quantifier-free formula $\varphi$, let $\psi$ be a quantifier-free formula over the universally quantified trace variables  $\pi_1,\pi_2 \ldots \pi_n$ and let $p$ be a fresh atomic proposition.
  It holds that $\tsys \models \forall \pi_1 \forall \pi_2 \ldots \forall \pi_n.\ \exists \pi_1' \exists \pi_2' \ldots\exists \pi_m'.\  \varphi$
  if, and only if, $\tsys^{\set{p}} \models \forall \pi_1 \forall \pi_2 \ldots \forall \pi_n.\ \exists \pi_1' \exists \pi_2' \ldots\exists \pi_m'. \Globally (p_{\pi_1} \leftrightarrow \psi) \rightarrow \varphi.$
  \label{theo:prophecy}
  \end{theorem}

Note that $\psi$ is restricted to refer only to \emph{universally} quantified trace variables. Without this restriction, the method would not be sound. In our example, $\psi = a_{\pi'}$ would
lead to the modified formula $\forall \pi \exists \pi'. (p_{\pi} \leftrightarrow a_{\pi'})\ \rightarrow\ (\Next a_\pi \leftrightarrow a_{\pi'})$, which could be satisfied with the strategy
that assigns $a_{\pi'}$ to $\true$ iff $p_\pi$ is $\false$, and thus falsifies the assumption that the prediction is correct, rather than ensuring that the original formula is true.

\begin{proof}
It is easy to see that the original specification implies the modified specification, since the original formula is the conclusion of the implication.
Assume that the modified specification holds. 
Since the prophecy variable $p$ is a fresh atomic proposition, and $\psi$ does not refer to the existentially chosen traces, we can, for every choice of the universally quantified traces, always choose the value of $p$ such that it guesses correctly, i.e., that $p$ is true whenever $\psi$ holds.
In this case, the conclusion and therefore the original specification must be true.
\qed
\end{proof}

Unfortunately, prophecy variables do not provide a complete proof technique. 
Consider a system allowing arbitrary sequences of $a$ and $b$ and this specification:
\begin{align*}
	\forall \pi \exists \pi'. & b_{\pi'} \land \Globally (b_{\pi'} \leftrightarrow \Next \neg b_{\pi'}) \\ 
	& \land (a_{\pi'} \rightarrow (a_\pi \WUntil (b_{\pi'} \land \neg a_\pi))) \\
	& \land (\neg a_{\pi'} \rightarrow (a_\pi \WUntil (\neg b_{\pi'} \land \neg a_\pi)))
\end{align*}
Intuitively, $\pi'$ has to be able to predict whether $\pi$ will stop outputting $a$ at an even or odd position of the trace. 
There is no HyperLTL formula to be used as $\psi$ in Theorem~\ref{theo:prophecy}, because, like LTL, HyperLTL can only express non-counting properties.
It is worth noting that in our practical experiments, the incompleteness was never a problem. 
In many cases, it is not even necessary to add prophecy variables at all. 
The presented proof technique is, thus, practically useful despite this incompleteness result. 

%-------------------------------------------------------------------------------
\subsection{Model Checking with Synthesized Strategies}
\label{sec:strategysynthesis}
%-------------------------------------------------------------------------------

We now extend the model checking approach with the automatic synthesis of the strategies for the existential quantifiers.
For a given $\hyperltl$ formula of the form $\forall^n \exists^m \varphi$ and a transition system $\tsys$, we search for a transition system $\tsys_\exists = \tuple{X, x_0, \mu, l_\exists}$, where
$X$ is a set of states,
$x_0 \in X$ is the designated initial state,
$\mu \colon X \times \Upsilon^n \to X$ is the transition function, and
$l_\exists \colon X \to \Upsilon^m$ is the labeling function,
such that $\tsys^n \times (\tsys^m \ ||\ \tsys_\exists) \models \zip(\varphi)$.
(Since for formulas of the form $\exists^m \forall^n \varphi$ the problem only differs in the input of $\tsys_\exists$, we focus on $\forall \exists\, \hyperltl$.)

\newcommand\donotshow[1]{}
\donotshow{
\begin{figure}[t]
  \centering
  \subfloat[The strategy synthesis problem with given implementations]{
    \begin{tikzpicture}[->,>=stealth',shorten >=1pt,auto,node distance=1cm,semithick,scale=1,transform shape]
      \tikzstyle{black-box}=[state,rectangle,color=white,fill=black];
      \tikzstyle{white-box}=[state,rectangle];
      \node[white-box] (sys-n) {$\tsys^n$};
      \node[white-box,right=of sys-n] (sys-m) {$\tsys^m$};
      \node[black-box,above=of sys-m] (sys-exists) {$\tsys_\exists$};
      \node[state,above=of sys-n] (env) {env};

      \draw (env) edge node {$I^n$} (sys-n)
            (env) edge node {$I^n$} (sys-exists)
            (sys-exists) edge node {$I^m$} (sys-m)
            (sys-n) edge node {$O^n$} +(0,-1)
            (sys-m) edge node {$O^m$} +(0,-1)
            ;
    \end{tikzpicture}
    \label{fig:strat-synth}
  }\qquad
  \subfloat[The implementation synthesis problem with given strategy (white box)]{
    \begin{tikzpicture}[->,>=stealth',shorten >=1pt,auto,node distance=1cm,semithick,scale=1,transform shape]
      \tikzstyle{black-box}=[state,rectangle,color=white,fill=black];
      \tikzstyle{white-box}=[state,rectangle];
      \node[black-box] (sys-n) {$\tsys^n$};
      \node[black-box,right=of sys-n] (sys-m) {$\tsys^m$};
      \node[white-box,above=of sys-m] (sys-exists) {$\tsys_\exists$};
      \node[state,above=of sys-n] (env) {env};

      \draw (env) edge node {$I^n$} (sys-n)
            (env) edge node {$I^n$} (sys-exists)
            (sys-exists) edge node {$I^m$} (sys-m)
            (sys-n) edge node {$O^n$} +(0,-1)
            (sys-m) edge node {$O^m$} +(0,-1)
            ;
    \end{tikzpicture}
    \label{fig:impl-synth}
  }\qquad
  \subfloat[The synthesis problem for strategy and implementation]{
    \begin{tikzpicture}[->,>=stealth',shorten >=1pt,auto,node distance=1cm,semithick,scale=1,transform shape]
      \tikzstyle{black-box}=[state,rectangle,color=white,fill=black];
      \tikzstyle{white-box}=[state,rectangle];
      \node[black-box] (sys-n) {$\tsys^n$};
      \node[black-box,right=of sys-n] (sys-m) {$\tsys^m$};
      \node[black-box,above=of sys-m] (sys-exists) {$\tsys_\exists$};
      \node[state,above=of sys-n] (env) {env};

      \draw (env) edge node {$I^n$} (sys-n)
            (env) edge node {$I^n$} (sys-exists)
            (sys-exists) edge node {$I^m$} (sys-m)
            (sys-n) edge node {$O^n$} +(0,-1)
            (sys-m) edge node {$O^m$} +(0,-1)
            ;
    \end{tikzpicture}
    \label{fig:hyper-synth}
  }
  \caption{Representation of the different synthesis problems as a distributed synthesis problem, i.e., the problem of synthesizing implementations of the black-box processes such that the joint combination with the white-box processes (which have a given implementation), satisfy the specification.}
 \end{figure}
}

 \begin{theorem}  \label{thm:strategy-synthesis-decidable}
   The strategy realizability problem for $\forall^* \exists^*$ formulas is $\twoexptime$-complete.
 \end{theorem}
 \begin{proof}[Sketch]
  We reduce the strategy synthesis problem to the problem of synthesizing a distributed reactive system with a single black-box process.
  This problem is decidable~\cite{conf/lics/FinkbeinerS05} and can be solved in $\twoexptime$.
  The lower bound follows from the LTL realizability problem~\cite{conf/popl/PnueliR89}.\qed
 \end{proof}
 The decidability result implies that there is an upper bound on the size of $\tsys_\exists$ that is doubly exponential in $\varphi$.
 Thus, the bounded synthesis approach~\cite{journals/sttt/FinkbeinerS13} can be used to search for increasingly larger implementations, until a solution is found or the maximal bound is reached, yielding an efficient decision procedure for the strategy synthesis problem.
 In the following, we describe this approach in detail.

%-------------------------------------------------------------------------------
\subsubsection{Bounded Synthesis of Strategies.}
%-------------------------------------------------------------------------------

We transform the synthesis problem into an SMT constraint satisfaction problem, where we leave the representation of strategies uninterpreted and challenge the solver to provide an interpretation.
Given a $\hyperltl$ formula $\forall^n \exists^m \varphi$ where $\varphi$ is quantifier-free, the model checking is based on the product of the $n$-fold self composition of the transition system $\tsys$, the $m$-fold self-composition of $\tsys$ where the strategy $\tsys_\exists$ controls the inputs, and the universal co-B\"uchi automaton $\ucw_\varphi$ representing the language $\lang(\varphi)$ of $\varphi$.

For a quantifier-free HyperLTL formula $\varphi$, we construct the  universal co-B\"uchi automaton $\ucw_\varphi$ such that $\lang(\ucw_\varphi)$ is the set of words $w$ such that $\unzip(w) \models \varphi$, i.e., the tuple of traces satisfies $\varphi$. 
We get this automaton by dualizing the non-deterministic B\"uchi automaton for $\neg\psi$~\cite{conf/post/ClarksonFKMRS14}, i.e., changing the branching from non-deterministic to universal and the acceptance condition from B\"uchi to co-B\"uchi.
Hence, $\tsys$ satisfies a universal $\hyperltl$ formula $\forall \pi_1 \dots \forall \pi_n \ldot \varphi$ if the traces generated by the self-composition $\tsys^n$ are a subset of $\lang(\ucw_\varphi)$.

In more detail, the algorithm searches for a transition system $\tsys_\exists = \tuple{X,x_0,\mu,l_\exists}$ such that the run graph of $\tsys^n$, $\tsys^m \ ||\ \tsys_\exists$, and $\ucw_\varphi$, written $\tsys^n \times (\tsys^m \ ||\ \tsys_\exists) \times \ucw_\varphi$, is accepting.
Formally, given a $\Gamma$-labeled $\Upsilon$-transition system $\tsys = \tuple{\T,\st_0,\tau,l}$ and a universal co-B\"uchi automaton $\ucw_\varphi = \tuple{Q,q_0,\delta,F}$, where $\delta \colon Q \times \Upsilon^{n+m} \times \Gamma^{n+m} \to \pow{Q}$, the run graph $\tsys^n \times (\tsys^m \ ||\ \tsys_\exists) \times \ucw_\varphi$ is the directed graph $(V,E)$, with the set of vertices $V = \T^n \times \T^m \times X \times Q$, initial vertex $v_\mathit{init} = ((\st_0,\dots,\st_0),(\st_0,\dots,\st_0),x_0,q_0)$ and the edge relation $E \subseteq V \times V$ satisfying $((\vec{\st_n},\vec{\st_m},x,q),(\vec{\st_n'},\vec{\st_m'},x',q')) \in E$ if, and only if
\begin{align*}
  & \exists \vec{\upsilon} \in \Upsilon^n \ldot \enspace
  \left( \vec{\st_n} \xrightarrow[\tau_n]{\vec{\upsilon}} \vec{\st_n'} \right) \land
  \left( \vec{\st_m} \xrightarrow[\tau_m]{l_\exists(x)} \vec{\st_m'} \right) \land
  \left( x \xrightarrow[\mu]{\vec{\upsilon}} x' \right) \\
  & \quad {} \quad {} \quad {} \quad {} \quad {} 
  \land
  q' \in \delta(q, \vec{\upsilon} \cdot l_\exists(x), l_n(\vec{\st_n}) \cdot l_m(\vec{\st_m}) ).
\end{align*}

\begin{theorem}  \label{thm:soundness-strategy-synthesis}
  Given $\tsys$, $\tsys_\exists$, and a $\hyperltl$ formula $\forall^n\exists^m \varphi$ where $\varphi$ is quantifier-free.
  Let $\ucw_\varphi$ be the universal co-B\"uchi automaton for $\varphi$. 
  If the run graph $\tsys^n \times (\tsys^m \ ||\ \tsys_\exists) \times \ucw_\varphi$ is accepting, then $\tsys \models \forall^n\exists^m \varphi$.
\end{theorem}
\begin{proof}
  Follows from Theorem~\ref{theo:stratex} and the fact that $\ucw_\varphi$ represents $\lang(\varphi)$. \qed
\end{proof}

The acceptance of a run graph is witnessed by an annotation $\lambda \colon V \rightarrow \nat \cup \set{\bot}$ which is a function mapping every reachable vertex $v \in V$ in the run graph to a natural number $\lambda(v)$, i.e., $\lambda(v) \neq \bot$.
Intuitively, $\lambda(v)$ returns the number of visits to rejecting states on any path from the initial vertex $v_\mathit{init}$ to $v$.
If we can bound this number for every reachable vertex, the annotation is \emph{valid} and the run graph is accepting.
Formally, an annotation $\lambda$ is valid, if
(1) the initial state is reachable ($\lambda(v_\mathit{init}) \neq \bot$) and
(2) for every $(v,v') \in E$ with $\lambda(v) \neq \bot$ it holds that $\lambda(v') \neq \bot$ and $\lambda(v) \vartriangleright \lambda(v')$ where $\trianglerighteq$ is $>$ if $v'$ is rejecting and $\geq$ otherwise.
Such an annotation exists if, and only if, the run graph is accepting~\cite{journals/sttt/FinkbeinerS13}.

We encode the search for $\tsys_\exists$ and the annotation $\lambda$ as an SMT constraint system.
Therefore, we use uninterpreted function symbols to encode $\tsys_\exists$ and $\lambda$.
A transition system $\tsys$ is represented in the constraint system by two functions, the transition function $\tau \colon \T \times \Upsilon \rightarrow \T$ and the labeling function $l \colon \T \rightarrow \Gamma$.
The annotation is split into two parts, a reachability constraint $\lambda^\bool \colon V \rightarrow \bool$ indicating whether a state in the run graph is reachable and a counter $\lambda^\# \colon V \rightarrow \nat$ that maps every reachable vertex $v$ to the maximal number of rejecting states $\lambda^\#(v)$ visited by any path from the initial vertex to $v$.
The resulting constraint asserts that there is a transition system $\tsys_\exists$ with an accepting run graph.
Note, that the functions representing the system $\tsys$ ($\tau \colon \T \times \Upsilon \rightarrow \T$ and $l \colon \T \to \Gamma$) are given, that is, they are interpreted. 
\begin{align*}
  &\exists \lambda^\bool \colon \T^n \times \T^m \times X \times Q \rightarrow \bool \ldot
  \exists \lambda^\nat \colon \T^n \times \T^m \times X \times Q \rightarrow \nat \ldot \\
  & \exists \mu \colon X \times \Upsilon^n \to X \ldot
  \exists l_\exists \colon X \to \Upsilon^m \\
  &\forall \vec{\upsilon} \in \Upsilon^n \ldot
  \forall \vec{\st_n},\vec{\st_n'} \in \T^{n} \ldot
  \forall \vec{\st_m},\vec{\st_m'} \in \T^{m} \ldot
  \forall q,q' \in Q \ldot
  \forall x,x' \in X \ldot \\
  & \lambda^\bool((\st_0,\dots,\st_0), (\st_0,\dots,\st_0), x_0, q_0) \land  {} \\
  &\Big( \lambda^\bool(\vec{\st_n},\vec{\st_m},x,q) \land q' \in \delta(q, (\vec{\upsilon} \cdot l_\exists(x)), (l \circ (\vec{\st_n} \cdot \vec{\st_m}))) \land x' = \mu(x, \vec{\upsilon}) \\
  & \quad {} \land
    \vec{\st_n'} = \tau_n(\vec{\st_n}, \vec{\upsilon})
    \land
    \vec{\st_m'} = \tau_m(\vec{\st_m}, l_\exists(x))
    \Big) \\
  &\Rightarrow \lambda^\bool(\vec{\st_n'}, \vec{\st_m'}, x', q') \land \lambda^\nat(\vec{\st_n}, \vec{\st_m}, x,q) \trianglerighteq \lambda^\nat(\vec{\st_n'}, \vec{\st_m'}, x',q')
\end{align*}
where $\trianglerighteq$ is $>$ if $q' \in F$ and $\geq$ otherwise.
The \emph{bounded synthesis algorithm} increases the bound of the strategy $\tsys_\exists$ until either the constraints system becomes satisfiable, or a given upper bound is reached.
In the case the constraint system is satisfiable, we can extract interpretations for the functions $\mu$ and $l_\exists$ using a solver that is able to produce models.
These functions then represent the synthesized transition system $\tsys_\exists$.

\begin{corollary}
  Given $\tsys$ and a $\hyperltl$ formula $\forall^*\exists^* \varphi$ where $\varphi$ is quantifier-free.
  If the constraint system is satisfiable for some bound on the size of $\tsys_\exists$ then $\tsys \models \forall^*\exists^* \varphi$.
\end{corollary}
\begin{proof}
  Follows immediately by Theorem~\ref{thm:soundness-strategy-synthesis}.\qed
\end{proof}
As the decision problem is decidable, we know that there is an upper bound on the size of a realizing $\tsys_\exists$ and, thus, the bounded synthesis approach is a decision procedure for the strategy realizability problem.
\begin{corollary}
  The bounded synthesis algorithm decides the strategy realizability problem for $\forall^* \exists^*\, \hyperltl$.
\end{corollary}
\begin{proof}
  The existence of such an upper bound follows from Theorem~\ref{thm:strategy-synthesis-decidable}.\qed
\end{proof}

%-------------------------------------------------------------------------------
\subsubsection{Approximating Prophecy.}
\label{sec:lookahead-synthesis}
%-------------------------------------------------------------------------------

We introduce a new parameter to the strategy synthesis problem to approximate the information about the future that can be captured using prophecy variables. 
This bound represents a constant \emph{lookahead} into future choices made by the environment. 
In other words, for a given $k \geq 0$, the strategy $\tsys_\exists$ is allowed to depend on choices of the $\forall$-player in the next $k$ steps. While constant lookahead is only an approximation of infinite clairvoyance, it suffices for many practical situations as shown by prior case studies~\cite{conf/cav/FinkbeinerRS15,conf/esop/DArgenioBBFH17}.

We present a solution to synthesizing transition systems with constant lookahead for $k \geq 0$ using bounded synthesis.
To simplify the presentation, we present the stand-alone problem with respect to a specification given as a universal co-B\"uchi automaton.
The integration into the constraint system for the $\forall^* \exists^*\, \hyperltl$ synthesis as presented in the previous section is then straightforward.
First, we present an extension to the transition system model that incorporates the notion of constant lookahead.
The idea of this extension is to replace the initial state $\st_0$ by a function $\mathit{init} \colon \Upsilon^k \to \T$ that maps input sequences of length $k$ to some state.
Thus, the transition system observes the first $k$ inputs, chooses some initial state based on those inputs, and then progresses with the same pace as the input sequence.
Next, we define the run graph of such a system $\tsys_k = \tuple{S,\mathit{init},\tau,l}$ and an automaton $\ucw = \tuple{Q,q_0,\delta,F}$, where $\delta \colon Q \times \Upsilon \times \Gamma \to Q$, as the directed graph $(V,E)$ with the set of vertices $V = \T \times Q \times \Upsilon^k$, the initial vertices $(\st,q_0,\vec{\upsilon}) \in V$ such that $\st = \mathit{init}(\vec{\upsilon})$ for every $\vec{\upsilon} \in \Upsilon^k$, and the edge relation $E \subseteq V \times V$ satisfying $( (\st,q,\upsilon_1 \upsilon_2 \cdots \upsilon_k), (\st',q',\upsilon'_1 \upsilon'_2 \cdots \upsilon'_k) ) \in E$ if, and only if
\begin{align*}
  \exists \upsilon_{k+1} \in \Upsilon \ldot
  \st \xrightarrow{\upsilon_{k+1}} \st' \land q' \in \delta(q, \upsilon_1, l(\st)) \land \bigwedge_{1 \leq i \leq k} \upsilon'_i = \upsilon_{i+1} .
\end{align*}

\begin{lemma}
  Given a universal co-B\"uchi automaton $\ucw$ and a $k$-lookahead transition system $\tsys_k$.
  $\tsys_k \models \ucw$ if, and only if, the run graph $\tsys_k \times \ucw$ is accepting.
\end{lemma}
Finally, synthesis amounts to solving the following constraint system:
\begin{align*}
  &\exists \lambda^\bool \colon \T \times Q \times \Upsilon^k \to \bool \ldot
  \exists \lambda^\nat \colon \T \times Q \times \Upsilon^k \to \nat \ldot \\
  & \exists \mathit{init} \colon \Upsilon^k \to \T \ldot
  \exists \tau \colon \T \times \Upsilon \rightarrow \T \ldot
  \exists l \colon \T \to \Gamma \ldot\\
  & (\forall \vec{\upsilon} \in \Upsilon^k \ldot \lambda^\bool(\mathit{init}(\vec{\upsilon}), q_0, \vec{\upsilon})) \land {}\\
  &\forall \upsilon_1\upsilon_2\cdots\upsilon_{k+1} \in \Upsilon^{k+1} \ldot
  \forall \st,\st' \in \T \ldot
  \forall q,q' \in Q \ldot  \\
  &\left( \lambda^\bool(\st,q,\upsilon_1\cdots\upsilon_k) \land \st' = \tau(\st, \upsilon_{k+1}) \land q' \in \delta(q, \upsilon_1, l(\st)) \right) \\
  &\Rightarrow \lambda^\bool(\st', q',\upsilon_2 \cdots \upsilon_{k+1}) \land \lambda^\nat(\st,q,\upsilon_1 \cdots \upsilon_k) \trianglerighteq \lambda^\nat(\st',q',\upsilon_2 \cdots \upsilon_{k+1})
\end{align*}

\begin{corollary}
  Given some $k \geq 0$, if the constraint system is satisfiable for some bound on the size of $\tsys_k$ then $\tsys_k \models \ucw$.
\end{corollary} 

%-------------------------------------------------------------------------------
\section{Synthesis with Quantifier Alternations}
\label{sec:fullsynthesis}
%-------------------------------------------------------------------------------

We now build on the introduced techniques to solve the \emph{synthesis} problem for $\hyperltl$ with quantifier alternation, that is, we search for implementations that satisfy the given properties.
%
%-------------------------------------------------------------------------------
\label{sec:ea-synthesis}
%-------------------------------------------------------------------------------
%
In previous work~\cite{conf/cav/FinkbeinerHLST18}, the synthesis problem for $\exists^*\forall^*\, \hyperltl$  was solved by a reduction to the distributed synthesis problem.
We present an alternative synthesis procedure that (1) introduces the necessary concepts for the synthesis of the $\forall^* \exists^*$ fragment and that (2) strictly decomposes the choice of the existential trace quantifier from the implementation.

Fix a formula of the form $\exists^m \forall^n \varphi$. 
We again reduce the verification problem to the problem of determining whether a run graph is accepting.
As the existential quantifiers do not depend on the universal ones, there is no future dependency and thus no need for prophecy variables or bounded lookahead.
Formally, $\tsys_\exists$ is a tuple $\tuple{X,x_0,\mu,l_\exists}$ such that 
$X$ is a set of states,
$x_0 \in X$ is the designated initial state,
$\mu \colon X \to X$ is the transition function, and
$l_\exists \colon X \to \Upsilon^m$ is the labeling function.
$\tsys_\exists$ produces infinite sequences of $(\Upsilon^m)^\omega$, without having any knowledge about the behavior of the universally quantified traces.
The run graph is then $(\tsys^m \ ||\ \tsys_\exists) \times \tsys^n \times \ucw_\varphi$.
The constraint system is built analogously to Sec.~\ref{sec:strategysynthesis}, with the difference that the representation of the system $\tsys$ is now also uninterpreted.
In the resulting SMT constraint system, we have two bounds, one for the size of the implementation $\tsys$ and one for the size of $\tsys_\exists$.

\begin{corollary}
  The bounded synthesis algorithm decides the realizability problem for $\exists^*\forall^1\, \hyperltl$ and is a semi-decision procedure for 
  $\exists^*\forall^{>1}\, \hyperltl$.
\end{corollary}
The synthesis problem for formulas in the $\forall^* \exists^*\, \hyperltl$ fragment uses the same reduction to a constraint system as the strategy synthesis in Sec.~\ref{sec:strategysynthesis}, with the only difference that the transition system $\tsys$ itself is uninterpreted.
In the resulting SMT constraint systems, we have three bounds, the size of the implementation $\tsys$, the size of the strategy $\tsys_\exists$, and the lookahead $k$.

\begin{corollary}
  Given a $\hyperltl$ formula $\forall^n\exists^m \varphi$ where $\varphi$ is quantifier-free.
  $\forall^n\exists^m \varphi$ is realizable if the SMT constraint system corresponding to the run graph $\tsys^n \times (\tsys^m \ ||\ \tsys_\exists) \times \ucw_\varphi$ is satisfiable for some bounds on $\tsys$, $\tsys_\exists$, and lookahead $k$.
\end{corollary}

%===============================================================================
\section{Implementations and Experimental Evaluation}
\label{sec:experiments}
%===============================================================================

We have integrated the model checking technique with a manually provided strategy into the HyperLTL hardware model checker $\mc$\footnote{Try the online tool interface with the latest version of $\mc$: \url{https://www.react.uni-saarland.de/tools/online/MCHyper/}}.
For the synthesis of strategies and reactive systems from hyperproperties, we have developed a separate bounded synthesis tool based on SMT-solving. 
In the following, we describe these implementations and report on experimental results. 
All experiments ran on a machine with dual-core Core i7, 3.3 GHz, and 16 GB memory.

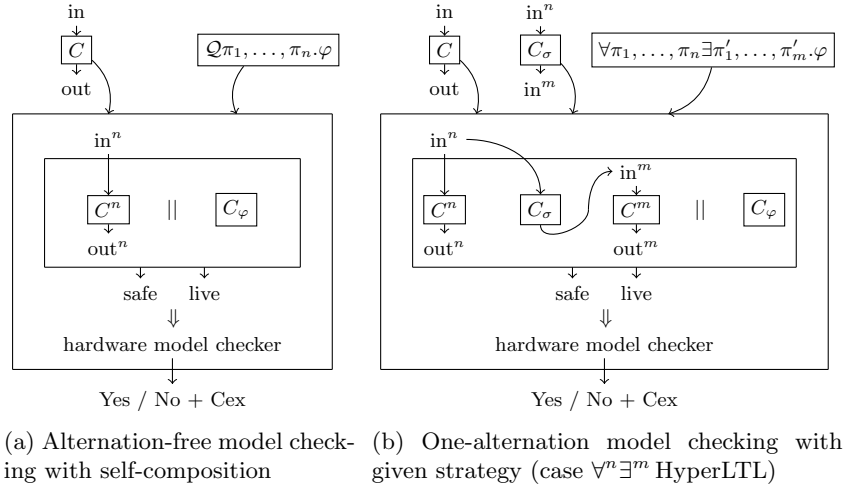
\begin{figure}[t]
  \centering
  \subfloat[Alternation-free model checking with self-composition]{
    \begin{tikzpicture}[transform shape,scale=0.85]
      \node[draw] at (-1.5,0) (C) {$C$};
      \node at (-1.5,0.6) (inC) {$\text{in}$};
      \node at (-1.5,-0.6) (outC) {$\text{out}$};
      \draw (inC) edge[->] (C);
      \draw (C) edge[->] (outC);
      \draw (C) edge[->,bend left] (-1,-1);
      
      \node[draw] at (1.5,0) (F) {$\mathcal{Q} \pi_1, \dots, \pi_n. \varphi$}; 
      \draw (F) edge[->,bend right] (1,-1);
         
      \draw (-2.5,-1) edge (2.5,-1);
      \node at (1.6,-1.2) (mc) {}; 
      \draw (-2.5,-1) edge (-2.5,-5);
      \draw (2.5,-1) edge (2.5,-5);
      \draw (-2.5,-5) edge (2.5,-5);
  
      \node[draw] at (-1,-2.5) (Cn) {$C^n$};
      \node at (-1,-1.4) (inCn) {$\text{in}^n$};
      \node at (-1,-3.1) (outCn) {$\text{out}^n$};
      \draw (inCn) edge[->] (Cn);
      \draw (Cn) edge[->] (outCn);
      
      \node at (0,-2.5) (P) {$||$};
      \node[draw] at (1,-2.5) (CF) {$C_\varphi$};
      
      \draw (-2,-1.7) edge (2,-1.7);
      \draw (-2,-1.7) edge (-2,-3.4);
      \draw (2,-1.7) edge (2,-3.4);
      \draw (-2,-3.4) edge (2,-3.4);
      
      \node at (-0.5,-3.8) (safe) {$\text{safe}$};
      \draw (-0.5,-3.4) edge[->] (safe);
      
      \node at (0.5,-3.8) (live) {$\text{live}$};
      \draw (0.5,-3.4) edge[->] (live);
          
      \node[rotate=-90] at (0,-4.2) (B) {$\Rightarrow$};
      \node at (0,-4.6) (abc) {hardware model checker};
      
      \node at (0,-5.5) (res) {$\text{Yes / No + Cex}$};
      \draw (abc) edge[->] (res);
      
    \end{tikzpicture}
    \label{fig:oldMCHyper}
  }
  \ \ 
  \subfloat[One-alternation model checking with given strategy (case $\forall^n \exists^m \, \hyperltl$)]{
    \begin{tikzpicture}[transform shape,scale=0.85]
      \node[draw] at (-1.5,0) (C) {$C$};
      \node at (-1.5,0.6) (inC) {$\text{in}$};
      \node at (-1.5,-0.6) (outC) {$\text{out}$};
      \draw (inC) edge[->] (C);
      \draw (C) edge[->] (outC);
      \draw (C) edge[->,bend left] (-1,-1);
      
      \node[draw] at (0,0) (S) {$C_\sigma$};
      \node at (0,0.6) (inS) {$\text{in}^n$};
      \node at (0,-0.6) (outS) {$\text{in}^m$};
      \draw (inS) edge[->] (S);
      \draw (S) edge[->] (outS);
      \draw (S) edge[->,bend left] (0.5,-1);
      
      \node[draw] at (2.7,0) (F) {$\forall \pi_1, \dots, \pi_n \exists \pi'_1, \dots, \pi'_m. \varphi$};       
      \draw (F) edge[->,bend left] (2,-1);
      
      \draw (-2.5,-1) edge (4.5,-1);
      \node at (3.6,-1.2) (mc) {}; 
      \draw (-2.5,-1) edge (-2.5,-5);
      \draw (4.5,-1) edge (4.5,-5);
      \draw (-2.5,-5) edge (4.5,-5);
      
      \node[draw] at (-1.5,-2.5) (Cn) {$C^n$};
      \node at (-1.5,-1.4) (inCn) {$\text{in}^n$};
      \node at (-1.5,-3.1) (outCn) {$\text{out}^n$};
      \draw (inCn) edge[->] (Cn);
      \draw (Cn) edge[->] (outCn);
      
      \node[draw] at (1.5,-2.5) (Cm) {$C^m$};
      \node at (1.5,-1.9) (inCm) {$\text{in}^m$};
      \node at (1.5,-3.1) (outCm) {$\text{out}^m$};
      \draw (inCm) edge[->] (Cm);
      \draw (Cm) edge[->] (outCm);
      
      \node[draw] at (0,-2.5) (Cs) {$C_\sigma$};
      \draw (inCn) edge[->,out=east,in=north] (Cs);
      \draw (Cs.south) edge[out=south,in=south] (0.75,-2.5);
      \draw (0.75,-2.5) edge[->,out=north,in=west] (inCm.west);
      
      \node at (2.5,-2.5) (P) {$||$};
      \node[draw] at (3.5,-2.5) (CF) {$C_\varphi$};
      
      \draw (-2,-1.7) edge (4,-1.7);
      \draw (-2,-1.7) edge (-2,-3.4);
      \draw (4,-1.7) edge (4,-3.4);
      \draw (-2,-3.4) edge (4,-3.4);
      
      \node at (0.5,-3.8) (safe) {$\text{safe}$};
      \draw (0.5,-3.4) edge[->] (safe);
      
      \node at (1.5,-3.8) (live) {$\text{live}$};
      \draw (1.5,-3.4) edge[->] (live);
      
      \node[rotate=-90] at (1,-4.2) (B) {$\Rightarrow$};
      \node at (1,-4.6) (abc) {hardware model checker};
      
      \node at (1,-5.5) (res) {$\text{Yes / No + Cex}$};
      \draw (abc) edge[->] (res);
      
    \end{tikzpicture}
    \label{fig:newMCHyper}
  }
  \caption{HyperLTL model checking with $\mc$}
  \vspace{-10pt}
\end{figure}

%===============================================================================
\subsubsection{Hardware Model Checking with Given Strategies.}
\label{sec:Impl:MC}
%===============================================================================
%
We have extended the model checker $\mc$~\cite{conf/cav/FinkbeinerRS15} from the alternation-free fragment to formulas with one quantifier alternation. 
The input to $\mc$ is a circuit description as an And-Inverter-Graph in the $\aiger$ format and a HyperLTL formula. 
Figure~\ref{fig:oldMCHyper} and Figure~\ref{fig:newMCHyper} show the model checking process in $\mc$ without and with quantifier alternation, respectively.
For formulas with quantifier alternation, the model checker now also accepts a strategy as an additional $\aiger$ circuit $C_\sigma$.
Based on this strategy, $\mc$ creates a new circuit where only the inputs of the universal system copies are exposed and the inputs of the existential system copies are determined by the strategy.
The new circuit is then model checked as described in~\cite{conf/cav/FinkbeinerRS15} with $\abc$~\cite{conf/cav/BraytonM10}.

We evaluate our extension of $\mc$ on formulas with quantifier alternation based on benchmarks from software doping~\cite{conf/esop/DArgenioBBFH17} and symmetry in mutual exclusion algorithms~\cite{conf/cav/FinkbeinerRS15}.  
Both considered problems have previously been analyzed with $\mc$; however, since the properties in both problems require quantifier alternation, we were previously limited to a (manually obtained) approximation of the properties as universal formulas. 
The correctness of manual approximations is not given but has to be shown separately. 
By directly model checking the formula with quantifier alternation we know that we are checking the correct formula without needing any additional proof of correctness. 

\noindent
\emph{Software Doping. }
D'Argenio et al.~\cite{conf/esop/DArgenioBBFH17} examined a clean and a doped version of an emission control program of a car and used the previous version of $\mc$ to formally verify approximations of these properties.  
Robust cleanness is expressed in the one-alternation fragment using two $\forall^2 \exists^1\, \hyperltl$ formulas (given in Prop. 19 in~\cite{conf/esop/DArgenioBBFH17}, cf. Sec.~\ref{sec:introduction}). 
In~\cite{conf/esop/DArgenioBBFH17}, the formulas were strengthened into alternation-free formulas that imply the original properties. 
Despite the quantifier alternation, Table~\ref{tbl:modelchecking} shows that the new version of $\mc$ verifies the precise formulas in roughly the same time as the alternation-free approximations~\cite{conf/esop/DArgenioBBFH17} while giving stronger correctness guarantees.
{\setlength{\tabcolsep}{6pt}
  \def\arraystretch{1.0}
\begin{table}[t]
  \caption{Experimental results for $\mc$ on the software doping and mutual exclusion benchmarks. All experiments used the IC3 option for $\abc$. Model and property names correspond to the ones used in~\cite{conf/esop/DArgenioBBFH17} and~\cite{conf/cav/FinkbeinerRS15}. 
  }
  \label{tbl:modelchecking}
  \centering
  \begin{tabular}{lccr}
    \hline\noalign{\smallskip}
    Model & \#Latches & Property & Time[s] \\
    \noalign{\smallskip}\hline\hline\noalign{\smallskip}
    EC 0.05    & 17 & $(10.a) + (10.b)$ & 1.8 \\
    EC 0.00625 & 23 & $(10.a) + (10.b)$ & 53.4 \\
    \noalign{\smallskip}\hline\noalign{\smallskip}
    AEC 0.05    & 19 & $(\neg 10.a) + (\neg10.b)$ & 2.8 \\
    AEC 0.00625 & 25 & $(\neg 10.a) + (\neg10.b)$ & 160.1 \\
    \noalign{\smallskip}\hline\hline\noalign{\smallskip}
    \multirow{2}{*}{Bakery.a.n.s}       & \multirow{2}{*}{47} & Sym5 & 50.6  \\
     							        & 					  & Sym6 & 27.5  \\
    \noalign{\smallskip}\hline\noalign{\smallskip}
    \multirow{2}{*}{Bakery.a.n.s.5proc} & \multirow{2}{*}{90} & Sym7 & 461.3 \\
     							        & 					  & Sym8 & 472.3 \\
    \noalign{\smallskip}\hline
  \end{tabular}
  \vspace{-10pt}
\end{table}}

\noindent
\emph{Symmetry in Mutual Exclusion Protocols. }
$\forall^*\exists^*\, \hyperltl$ allows us to specify symmetry for mutual exclusion protocols.
In such protocols, we wish to guarantee that every request is eventually answered, and the grants are mutually exclusive.
In our experiments, we used an implementation of the Bakery protocol~\cite{journals/cacm/Lamport74a}. 
Table~\ref{tbl:modelchecking} shows the verification results for the precise $\forall^1 \exists^1$ properties. 
Comparing these results to the performance on the approximations of the symmetry properties~\cite{conf/cav/FinkbeinerRS15}, we, again, observe that the verification times are similar. 
However, we gain the additional correctness guarantees as described above. 

%===============================================================================
\subsubsection{Strategy and System Synthesis.}
\label{sec:experimentsSynthesis}
%===============================================================================
%
For the synthesis of strategies for existential quantifiers and for the synthesis of reactive
systems from hyperproperties, we have developed a separate bounded synthesis tool based on SMT-solving with $\zthree$~\cite{conf/tacas/MouraB08}.
Our evaluation is based on two benchmark families, the \emph{dining cryptographers} problem~\cite{journals/cacm/Chaum85} and a simplified version of the symmetry problem in mutual exclusion protocols discussed previously.
The results are shown in Table~\ref{tbl:synthesis}.
Obviously, synthesis operates at a vastly smaller scale than model checking with given strategies.
In the dining cryptographers example, $\zthree$ was unable to find an implementation for the full synthesis problem, but could easily synthesize strategies for the existential trace quantifiers when provided with an implementation.
With the progress of constraint solver that employ quantification over Boolean functions~\cite{conf/sat/TentrupR19} we expect scalability improvements of our synthesis approach.

{\setlength{\tabcolsep}{2pt}
\def\arraystretch{1.0}
\begin{table}[t]
  \caption{Summary of the experimental results on the benchmarks sets described in Section~\ref{sec:experiments}.
    When no hyperproperty is given, only the LTL part is used.
  }
  \label{tbl:synthesis}
  \centering
  \begin{tabular}{llllr}
    \hline\noalign{\smallskip}
    Instance & Hyperproperty & $\card{\tsys}$ & $\card{\tsys_\exists}$ & Time[s]\\
    \noalign{\smallskip}\hline\hline\noalign{\smallskip}
    \multirow{2}{*}{Dining Cryptographers} 
       & distributed + deniability  & \multicolumn{2}{c}{TO}  \\
       & distributed + deniability with given $\tsys$ & (1)&1 & 1.2  \\
    \noalign{\smallskip}\hline\noalign{\smallskip}
    \multirow{2}{*}{Mutex} & ---  & 2 & -- & $<1$ \\
       & symmetry & 3 & 1 & 3.4 \\
    \noalign{\smallskip}\hline\noalign{\smallskip}
    \multirow{4}{*}{Mutex w/o spurious grants} & ---  & 3 & -- & $<1$ \\
       & symmetry & 3 & 1 & 3.9 \\
       & wait-free & 3 & 3 & 46 \\ 
       & symmetry + wait-free & 3 & 1+3 & 840 \\ 
    \noalign{\smallskip}\hline
  \end{tabular}
  \vspace{-10pt}
\end{table}}

%===============================================================================
\section{Conclusions}
\label{sec:conclusion}
%===============================================================================

We have presented model checking and synthesis techniques for hyperliveness properties
expressed as HyperLTL formulas with quantifier alternation. The alternation
makes it possible to specify hyperproperties such as generalized noninterference, symmetry, and deniability. 
Our approach is the first method for the
synthesis of reactive systems from HyperLTL formulas with quantifier alternation and the first
practical method for the verification of such specifications.

The approach is based on a game-theoretic view of existential
quantifiers, where the $\exists$-player reacts to decisions of the
$\forall$-player.  The key advantage is that 
the complementation of the system automaton is avoided
(cf. \cite{conf/cav/FinkbeinerRS15}). Instead, a strategy must be
found for the $\exists$-player.  Since this can be done either
manually or through automatic synthesis, the user of the model
checking or synthesis tool has the opportunity to trade some
automation for a significant gain in performance.

\medskip

\noindent{\bf Acknowledgements.}\;
We would like to thank Sebastian Biewer for providing the software doping models and formulas, 
Marvin Stenger for his advice on our synthesis experiments, 
and Jana Hofmann for her helpful comments on a draft of this paper. 

\bibliographystyle{splncs04}
\bibliography{main}

\begin{thebibliography}{10}
\providecommand{\url}[1]{\texttt{#1}}
\providecommand{\urlprefix}{URL }
\providecommand{\doi}[1]{https://doi.org/#1}

\bibitem{journals/tcs/AbadiL91}
Abadi, M., Lamport, L.: The existence of refinement mappings. Theor. Comput.
  Sci.  \textbf{82}(2),  253--284 (1991). \doi{10.1016/0304-3975(91)90224-P}

\bibitem{conf/lfcs/BartheCK13}
Barthe, G., Crespo, J.M., Kunz, C.: Beyond 2-safety: Asymmetric product
  programs for relational program verification. In: Proceedings of {LFCS}.
  LNCS, vol.~7734, pp. 29--43. Springer (2013).
  \doi{10.1007/978-3-642-35722-0\_3}

\bibitem{conf/csfw/BartheDR04}
Barthe, G., D'Argenio, P.R., Rezk, T.: Secure information flow by
  self-composition. In: Proceedings of {CSFW}. pp. 100--114. {IEEE} Computer
  Society (2004). \doi{10.1109/CSFW.2004.17}

\bibitem{conf/cav/BraytonM10}
Brayton, R.K., Mishchenko, A.: {ABC:} an academic industrial-strength
  verification tool. In: Proceedings of {CAV}. LNCS, vol.~6174, pp. 24--40.
  Springer (2010). \doi{10.1007/978-3-642-14295-6\_5}

\bibitem{journals/cacm/Chaum85}
Chaum, D.: Security without identification: Transaction systems to make big
  brother obsolete. Commun. {ACM}  \textbf{28}(10),  1030--1044 (1985).
  \doi{10.1145/4372.4373}

\bibitem{conf/post/ClarksonFKMRS14}
Clarkson, M.R., Finkbeiner, B., Koleini, M., Micinski, K.K., Rabe, M.N.,
  S{\'{a}}nchez, C.: Temporal logics for hyperproperties. In: Proceedings of
  {POST}. LNCS, vol.~8414, pp. 265--284. Springer (2014).
  \doi{10.1007/978-3-642-54792-8\_15}

\bibitem{journals/jcs/ClarksonS10}
Clarkson, M.R., Schneider, F.B.: Hyperproperties. Journal of Computer Security
  \textbf{18}(6),  1157--1210 (2010). \doi{10.3233/JCS-2009-0393}

\bibitem{conf/cav/CookKP15}
Cook, B., Khlaaf, H., Piterman, N.: On automation of {CTL}* verification for
  infinite-state systems. In: Proceedings of {CAV}. LNCS, vol.~9206, pp.
  13--29. Springer (2015). \doi{10.1007/978-3-319-21690-4\_2}

\bibitem{conf/esop/DArgenioBBFH17}
D'Argenio, P.R., Barthe, G., Biewer, S., Finkbeiner, B., Hermanns, H.: Is your
  software on dope? - formal analysis of surreptitiously "enhanced" programs.
  In: Proceedings of {ESOP}. LNCS, vol. 10201, pp. 83--110. Springer (2017).
  \doi{10.1007/978-3-662-54434-1\_4}

\bibitem{journals/jcs/DSouzaHRS11}
D'Souza, D., Holla, R., Raghavendra, K.R., Sprick, B.: Model-checking
  trace-based information flow properties. Journal of Computer Security
  \textbf{19}(1),  101--138 (2011). \doi{10.3233/JCS-2010-0400}

\bibitem{conf/concur/FinkbeinerH16}
Finkbeiner, B., Hahn, C.: Deciding hyperproperties. In: Proceedings of
  {CONCUR}. LIPIcs, vol.~59, pp. 13:1--13:14. Schloss Dagstuhl -
  Leibniz-Zentrum fuer Informatik (2016). \doi{10.4230/LIPIcs.CONCUR.2016.13}

\bibitem{conf/atva/FinkbeinerHH18}
Finkbeiner, B., Hahn, C., Hans, T.: {MGHyper}: Checking satisfiability of
  {HyperLTL} formulas beyond the $\exists^* \forall^*$ fragment. In:
  Proceedings of {ATVA}. LNCS, vol. 11138, pp. 521--527. Springer (2018).
  \doi{10.1007/978-3-030-01090-4\_31}

\bibitem{conf/cav/FinkbeinerHLST18}
Finkbeiner, B., Hahn, C., Lukert, P., Stenger, M., Tentrup, L.: Synthesizing
  reactive systems from hyperproperties. In: Proceedings of {CAV}. LNCS, vol.
  10981, pp. 289--306. Springer (2018). \doi{10.1007/978-3-319-96145-3\_16}

\bibitem{conf/cav/FinkbeinerHS17}
Finkbeiner, B., Hahn, C., Stenger, M.: {EAHyper}: Satisfiability, implication,
  and equivalence checking of hyperproperties. In: Proceedings of {CAV}. LNCS,
  vol. 10427, pp. 564--570. Springer (2017).
  \doi{10.1007/978-3-319-63390-9\_29}

\bibitem{conf/rv/FinkbeinerHST17}
Finkbeiner, B., Hahn, C., Stenger, M., Tentrup, L.: Monitoring hyperproperties.
  In: Proceedings of {RV}. LNCS, vol. 10548, pp. 190--207. Springer (2017).
  \doi{10.1007/978-3-319-67531-2\_12}

\bibitem{conf/tacas/FinkbeinerHST18}
Finkbeiner, B., Hahn, C., Stenger, M., Tentrup, L.: {RVHyper}: {A} runtime
  verification tool for temporal hyperproperties. In: Proceedings of {TACAS}.
  LNCS, vol. 10806, pp. 194--200. Springer (2018).
  \doi{10.1007/978-3-319-89963-3\_11}

\bibitem{conf/cav/FinkbeinerHT18}
Finkbeiner, B., Hahn, C., Torfah, H.: Model checking quantitative
  hyperproperties. In: Proceedings of {CAV}. LNCS, vol. 10981, pp. 144--163.
  Springer (2018). \doi{10.1007/978-3-319-96145-3\_8}

\bibitem{conf/cav/FinkbeinerRS15}
Finkbeiner, B., Rabe, M.N., S{\'{a}}nchez, C.: Algorithms for model checking
  {HyperLTL} and {HyperCTL*}. In: Proceedings of {CAV}. LNCS, vol.~9206, pp.
  30--48. Springer (2015). \doi{10.1007/978-3-319-21690-4\_3}

\bibitem{conf/lics/FinkbeinerS05}
Finkbeiner, B., Schewe, S.: Uniform distributed synthesis. In: Proceedings of
  {LICS}. pp. 321--330. {IEEE} Computer Society (2005).
  \doi{10.1109/LICS.2005.53}

\bibitem{journals/sttt/FinkbeinerS13}
Finkbeiner, B., Schewe, S.: Bounded synthesis. {STTT}  \textbf{15}(5-6),
  519--539 (2013). \doi{10.1007/s10009-012-0228-z}

\bibitem{conf/sp/GoguenM82a}
Goguen, J.A., Meseguer, J.: Security policies and security models. In:
  Proceedings of {S\&P}. pp. 11--20. {IEEE} Computer Society (1982).
  \doi{10.1109/SP.1982.10014}

\bibitem{conf/tacas/HahnST19}
Hahn, C., Stenger, M., Tentrup, L.: Constraint-based monitoring of
  hyperproperties. In: Proceedings of {TACAS}. LNCS, vol. 11428, pp. 115--131.
  Springer (2019). \doi{10.1007/978-3-030-17465-1\_7}

\bibitem{conf/csfw/HuismanWS06}
Huisman, M., Worah, P., Sunesen, K.: A temporal logic characterisation of
  observational determinism. In: Proceedings of {CSFW}. p.~3. {IEEE} Computer
  Society (2006). \doi{10.1109/CSFW.2006.6}

\bibitem{conf/icalp/Klein015}
Klein, F., Zimmermann, M.: How much lookahead is needed to win infinite games?
  In: Proceedings of {ICALP}. LNCS, vol.~9135, pp. 452--463. Springer (2015).
  \doi{10.1007/978-3-662-47666-6\_36}

\bibitem{journals/cacm/Lamport74a}
Lamport, L.: A new solution of {Dijkstra}'s concurrent programming problem.
  Commun. {ACM}  \textbf{17}(8),  453--455 (1974). \doi{10.1145/361082.361093}

\bibitem{journals/iandc/LynchV95}
Lynch, N.A., Vaandrager, F.W.: Forward and backward simulations: I. untimed
  systems. Inf. Comput.  \textbf{121}(2),  214--233 (1995).
  \doi{10.1006/inco.1995.1134}

\bibitem{conf/sp/McCullough88}
McCullough, D.: Noninterference and the composability of security properties.
  In: Proceedings of {S\&P}. pp. 177--186. {IEEE} Computer Society (1988).
  \doi{10.1109/SECPRI.1988.8110}

\bibitem{journals/entcs/MeydenZ07}
van~der Meyden, R., Zhang, C.: Algorithmic verification of noninterference
  properties. Electr. Notes Theor. Comput. Sci.  \textbf{168},  61--75 (2007).
  \doi{10.1016/j.entcs.2006.11.002}

\bibitem{conf/tacas/MouraB08}
de~Moura, L.M., Bj{\o}rner, N.: {Z3:} an efficient {SMT} solver. In:
  Proceedings of {TACAS}. LNCS, vol.~4963, pp. 337--340. Springer (2008).
  \doi{10.1007/978-3-540-78800-3\_24}

\bibitem{conf/popl/PnueliR89}
Pnueli, A., Rosner, R.: On the synthesis of a reactive module. In: Proceedings
  of {POPL}. pp. 179--190. {ACM} Press (1989). \doi{10.1145/75277.75293}

\bibitem{conf/sat/TentrupR19}
Tentrup, L., Rabe, M.N.: Clausal abstraction for {DQBF}. In: Proceedings of
  {SAT} (2019). \doi{10.1007/978-3-030-24258-9_27}

\end{thebibliography}

\end{document}